\documentclass{llncs}

\input{xy}
\xyoption{all}

\usepackage[italian,english]{babel}
\usepackage[T1]{fontenc}
\usepackage{ucs} 
\usepackage{graphicx,color}
\usepackage{amsmath,amssymb,amscd,amsfonts,mathtools}
\usepackage{enumerate}
\def\draft{1}   

\ifnum\draft=1 
    \def\ShowAuthNotes{1}
\else
    \def\ShowAuthNotes{0}
\fi

\ifnum\ShowAuthNotes=1
\newcommand{\authnote}[2]{{ \textbf{\textcolor{red}{[#1's Note: \em #2 ]} }}}
\else
\newcommand{\authnote}[2]{}
\fi

  {\left\lbrace\begin{array}{@{}l@{}}}%
  {\end{array}\right.}

\newcommand{\PL}{P_{\mathsf{L}}}

\newcommand{\PR}{P_{\mathsf{R}}}

\newcommand{\Refresh}{\mathsf{Refresh}}
\newcommand{\ExpRefresh}{\mathsf{ExpRefresh}}

\newcommand{\ReconstructRefresh}{\mathsf{ReconstructRefresh}}
\newcommand{\ExpReconstructRefresh}{\mathsf{ExpReconstructRefresh}}

\newcommand{\inp}[1]{\langle #1 \rangle}

\newcommand{\Enc}{\mathrm{Enc}}
\newcommand{\Dec}{\mathrm{Dec}}

\newcommand{\cpO}{\cO'}

\newcommand{\leak}{\lambda}
\newcommand{\bin}{\{0,1\}}

\newcommand{\draw}{\leftarrow}
\newcommand{\view}{\mathsf{view}}
\newcommand{\viewL}{\view_{\mathsf{L}}}
\newcommand{\viewR}{\view_{\mathsf{R}}}

\newcommand{\out}[1]{}

\newcommand{\negl}{negl}

\newcommand{\inter}{\rightleftarrows}

\newcommand{\cA}{{\cal A}}

\newcommand{\cL}{{\cal L}}
\newcommand{\cO}{{\cal O}}

\newcommand{\cR}{{\cal R}}

\newcommand{\FF}{\mathbb{F}}

\newcommand{\NN}{\mathbb{N}}

 \def\eqd{\,{\buildrel d \over =}\,}

\usepackage{cite}

\title{Efficient Refreshing Protocol for Leakage-Resilient Storage Based on the Inner-Product Extractor
\thanks{This work was supported by the {\em
      WELCOME/2010-4/2} grant founded within the framework of the {\em
      EU Innovative Economy (National Cohesion Strategy) Operational
      Programme}.}}

\author{Marcin Andrychowicz\inst{1}}

\institute{University of Warsaw}

\newcommand{\samp}{\leftarrow}

\newcommand{\Omegaretr}{\Omega}

\newcommand{\abs}[1]{\left|#1\right|}

\newcommand{\cM}{{\cal M}}

\newcommand{\dist}{\Delta}

\usepackage[table]{xcolor}

\newcommand{\FZ}{(\FF \setminus \{0\})}

\begin{document}

\maketitle

\begin{abstract}

A recent trend in cryptography is to protect data and computation against various side-channel attacks.
Dziembowski and Faust (TCC 2012) have proposed a general way to protect arbitrary circuits against any continual leakage
assuming that:
(i) the memory is divided into the parts, which leaks independently
(ii) the leakage in each observation is bounded
(iii) the circuit has an access to a leak-free component, which samples random orthogonal vectors.
The pivotal element of their construction is a protocol for \emph{refreshing} the so-called Leakage-Resilient Storage (LRS).

In this note, we present a more efficient and simpler protocol for refreshing LRS under the same assumptions.
Our solution needs $O(n)$ operations to fully refresh the secret (in comparison to $\Omega(n^2)$ for a protocol of Dziembowski and Faust), where $n$ is a security parameter that describes the maximal amount of leakage in each invocation of the refreshing procedure.
\end{abstract}

\section{Introduction}
A leakage-resilient cryptography has been intensively studied in the recent years
(cf.~for instance \cite{MR,DP08,FaustKPR10,GoldwasserR10,DBLP:conf/focs/DodisHLW10,DBLP:conf/focs/BrakerskiKKV10,Alison1,DF11,DF12,GR12}).
This note is based on a work by Dziembowski and Faust \cite{DF12}.
It follows the assumptions, construction and notation from the mentioned work.
We briefly review the settings, for the complete description of the model we refer the reader to \cite{DF12}.

We first start with the definition of the {\em Leakage-Resilient Storage (LRS)} \cite{DDV10},
which is a randomized encoding scheme
$(\Enc : \cM \rightarrow \cL \times \cR, \Dec: \cL \times \cR \rightarrow \cM)$,
resilient to leakage in the following sense. 
Let $m \in \cM$ be a message, and let $(l,r) :=\Enc(m)$.
Then, an adversary that learns some partial information $f(l)$ about $l$ and (independently) $g(r)$ about $r$ should gain no information about the encoded message $m$.  
The idea is to keep $l$ and $r$ on the different memory parts, which leak independently.
We will model that setting assuming that they are kept be different \emph{parties}, which
can perform computation and exchange messages.

More precisely (citing verbatim \cite{DF12}), for some $c, \ell, \leak \in \NN$ let $M_1,\ldots,M_\ell \in \bin^c$ denote the contents of the memory
parts, then we define a {\em $\leak$-leakage game}
played between an adaptive adversary $\cA$, called a {\em
  $\leak$-limited leakage adversary}, and a {\em leakage oracle}
$\Omegaretr(M_1,\ldots,M_\ell)$ as follows. For some $m \in \NN$, the adversary $\cA$ can adaptively issue a sequence  $\{(x_i,f_i)\}_{i=1}^m$ of requests to the oracle $\Omegaretr(M_1,\ldots,M_\ell)$, where $x_i \in \{1,\ldots,\ell\}$ and $f_i:\bin^c \rightarrow \bin^{\leak_i}$ with $\leak_i \leq \leak$. To each such a query the oracle replies with $f_i(M_{x_i})$ and we say that in this case the adversary $\cA$ {\em retrieved the value $f_i(M_{x_i})$ from $M_{x_i}$}. The only
restriction is that in total the adversary does not
retrieve more than  $\leak$ bits from
each memory part.
In the following, let $(\cA \inter (M_1,\ldots,M_\ell))$ be
the output of $\cA$ at the end of this game.

An LRS $\Phi$ is said to be $(\leak,\epsilon)$-secure, if for any $S,S' \in \cM$ and any $\leak$-limited adversary $\cA$, we have
$\dist(\cA \inter (L,R); \cA \inter (L',R')) \leq \epsilon$,
 where $(L,R) \draw \Enc(S)$ and $(L',R') \draw \Enc(S')$, for any two secrets $S, S' \in \cM$.

A variant of LRS $\Phi^{n}_\FF$ introduced in \cite{DF11} is based of the inner-product extractor.
A secret $S \in \FF$ (where $\FF$ is an arbitrary finite field) is encoded using two random vectors $L,R \in \FF^n$, such that
$S = \inp{L,R}$.
In this note we only allow the encodings such that $L,R \in \FZ^n$.
Moreover, we will assume that $F$ is fairly large in comparison to $n$, that is $|\FF| \ge 4n$.
Dziembowski and Faust \cite{DF12} showed the following lemma.
 \begin{lemma}\label{cor:par}
    Suppose $\abs{\FF} = \Omega(n)$.  Then, LRS $\Phi^{n}_\FF$ is $(0.49
    \cdot \log_2 \abs{\FF^n} -1, \negl(n))$-secure, for some negligible function $\negl$.
  \end{lemma}

Dziembowski and Faust \cite{DF12} have proposed a \emph{compiler}, which transforms arbitrary circuits over $\FF$
into functionally equivalent circuits secure
 against any continual leakage
assuming that:
\begin{enumerate}
\item the memory is divided into the parts, which leak independently,
\item the leakage from each memory part is bounded,
\item the circuit has an access to a leak-free component, which samples random orthogonal vectors.
\end{enumerate}
A pivotal point in the construction is the $\Refresh^n_\FF$
protocol, which refreshes the encoding of the secret.
It is run by two parties $P_L$ holding $L$ and $P_R$ holding $R$.
At the end of the protocol $P_L$ outputs $L'$ and $P_R$ outputs $R'$
such that $\inp{L,R} = \inp{L',R'}$ but except of this $(L',R')$ is uniform and independent
of $(L,R)$.

The only fact about $\Refresh^n_\FF$, which is used in the security proof presented in \cite{DF12}
is the existence of the \emph{reconstructor} procedure (an idea introduced earlier in\cite{FaustRRTV10}).
Informally, the reconstructor is a protocol that for inputs $(L,L')$
 held by $\PL$ and $(R,R')$ held by $\PR$ (where $L,L',R,R' \in \FZ^n$)
 such that $\inp{L,R} = \inp{L',R'}$ allows the parties to reconstruct
 the views that they would have in the $\Refresh_\FF^n(L,R)$ protocol,
 assuming that $(L',R')$ is an output of $\Refresh_\FF^n(L,R)$.

The $\Refresh^n_\FF$ protocol presented in \cite{DF12} performs $O(n^2)$ operations.
It is there used in a ,,generalized multiplication'' protocol as a sub-routine, what
leads at the end to $O(n^4)$ blow-up of the circuit's size while securing it against
leakages.
The protocol presented in this note needs $O(n)$ operations to refresh the secret,
what leads to $O(n^2)$ blow-up of the circuit's size.


\section{Leakage-Resilient Refreshing of LRS}\label{sec:refresh}
Similarly as in \cite{DF12} we assume
that the players have access to a leak-free component that samples
uniformly random pairs of orthogonal vectors. Technically, we will
assume that we have an oracle $\cpO$ that samples a uniformly random vector
$((A,\tilde{A}),(B,\tilde{B})) \in (\FF^n)^4$, subject to the constraint that the following three conditions hold:
\begin{enumerate}
\item \label{C1} $\inp{A,B} + \inp{\tilde{A},\tilde{B}}=0$,
\item \label{C2} $A_i \neq 0$ for $1 \le i \le n$,
\item \label{C3} $\tilde{B}_i \neq 0$ for $1 \le i \le n$.
\end{enumerate}
Note that although our oracle is slightly different from the oracle
$\cO$ used
in \cite{DF12}, it may be easily ,,simulated'' by the players having
access to $\cO$.

The refreshing scheme is presented in Figure~\ref{fig:refresh}.
The general idea behind the protocol is similar to one, which appeared
in \cite{DF12}.
Denote $\alpha :=
\inp{A,B} (= - \inp{\tilde{A},\tilde{B}})$.   The Steps \ref{step:R1}
and \ref{step:R2} are needed to refresh the share of $\PR$.  This is
done by generating, with the ``help'' of $(A,B)$ (coming from $\cpO$)
a vector $X$ such that
\begin{equation}
  \label{eq:inpLX}
\inp{L,X} = \alpha.
\end{equation}
The key difference between our approach and the protocol from \cite{DF12}
is a new and more efficient way of generating such $X$.
Eq.\ (\ref{eq:inpLX}) comes from a summation:
$\inp{L,X} = \sum_{i=1}^n L_i X_i = \sum_{i=1}^n L_i V_i B_i = \sum_{i=1}^n L_i L^{-1}_i A_i B_i =  \inp{A,B} = \alpha$.
Then, vector $X$ is added
to the share of $\PR$ by setting (in Step \ref{step:R2}) $R' := R +
X$.  Hence we get $\inp{L,R'} = \inp{L,R} + \inp{L,X} = \inp{L,R} +
\alpha$.  Symmetrically, in Steps \ref{step:L1} and \ref{step:L2} the
players refresh the share of $\PL$, by first generating $\tilde{X}$ such that
$\inp{\tilde{X},R'} = -\alpha$, and then setting $L' = L + \tilde{X}$.  By similar
reasoning as before, we get $\inp{L',R'} = \inp{L,R'}  - \alpha$,
which, in turn is equal to $\inp{L,R}$.  Hence,
$\inp{L,R} = \inp{L',R'}$.

\begin{center}
\begin{figure}
\fbox{ \footnotesize
 \begin{minipage}[c]{.95\linewidth}
   \begin{center} {\bf Protocol $(L',R') \draw \Refresh^n_\FF((L,R))$: } \end{center}
   \vspace{0.2cm}
\begin{minipage}[c]{.95\linewidth}
{\bf Input $(L,R)$: } $L \in \FZ^n$ is given to $\PL$ and $R \in \FZ^n$ is given to $\PR$.\\
 \vspace{-0.1cm}
\begin{enumerate}
      \item \label{step:0} Let $((A,\tilde{A}),(B,\tilde{B})) \samp \cpO$ and give $(A,\tilde{A})$ to $\PL$ and $(B,\tilde{B})$ to $\PR$.
   \begin{center} {\bf Refreshing the share of $\PR$:} \end{center}
      \item \label{step:R1} The player $\PL$ computes a vector $V$ such that $V_i := L_i^{-1} \cdot A_i$ for $1 \le i \le n$ and sends $V$ to $\PR$.
      \item \label{step:R2} The player $\PR$ computes a vector $X$ such that $X_i :=  V_i \cdot B_i$ for $1 \le i \le n$ and sets $R' := R + X$.
      \item \label{step:R3} If there exists $i$ such that $R'_i = 0$, then the protocol is restarted from the very beginning with the new vectors sampled from $\cpO$.
      \vspace{-0.1cm}
   \begin{center} {\bf Refreshing the share of $\PL$:} \end{center}
      \item \label{step:L1} The player $\PR$ computes a vector $\tilde{V}$ such that $\tilde{V}_i := {R'_i}^{-1} \cdot \tilde{B}_i$ for $1 \le i \le n$ and sends $\tilde{V}$ to $\PL$.
      \item \label{step:L2} The player $\PL$ computes a vector $\tilde{X}$ such that $\tilde{X}_i :=  \tilde{V}_i \cdot \tilde{A}_i$ for $1 \le i \le n$ and  sets $L' := L + \tilde{X}$.
      \item \label{step:L3} If there exists $i$ such that $L'_i = 0$, then the protocol is restarted from the very beginning with the new vectors sampled from $\cpO$.
\end{enumerate}
\vspace{-0.1cm}
{\bf Output:} The players output $(L',R')$.\\
{\bf Views:} The view $\viewL$ of player $\PL$ is $(L,A,V,\tilde{A},\tilde{V})$ and the view $\viewR$ of player $\PR$ is $(R,B,V,\tilde{B},\tilde{V})$.
\end{minipage}
\end{minipage}
}
\caption{Protocol $\Refresh^n_\FF$.  Oracle $\cpO$ samples random
  vectors $(A,\tilde{A},B,\tilde{B}) \in \FZ^n \times \FF^n \times \FF^n \times \FZ^n$ such that
  $\inp{A,B} = - \inp{\tilde{A},\tilde{B}}$.
  Note that the inverses in Steps \ref{step:R1} and \ref{step:L1}
  always exist, because $L, R \in \FZ^n$.
  Steps \ref{step:R3} and \ref{step:L3} guarantee that this condition
  is preserved under the execution of the protocol $\Refresh^n_\FF$.
  It can be easily proven that the protocol is restarted with
  a bounded probability regardless of $n$ (but keeping $|\FF| \ge 4n$), so it
  changes the efficiency of the algorithm only by a constant factor.
}\label{fig:refresh}
\end{figure}
\end{center}

 \section{Reconstructor for $\Refresh^n_\FF$}

 We now show a {\em reconstructor for the $\Refresh^n_\FF$
   protocol}.  Informally, the reconstructor is a protocol that for inputs $(L,L')$
 held by $\PL$ and $(R,R')$ held by $\PR$ (where $L,L',R,R' \in \FZ^n$)
 such that $\inp{L,R} = \inp{L',R'}$ allows the parties to reconstruct
 the views that they would have in the $\Refresh_\FF^n(L,R)$ protocol,
 assuming that $(L',R')$ is an output of $\Refresh_\FF^n(L,R)$.  The
 key feature of this reconstructor is that it does not require any
 interaction between the players.  The only ``common randomness'' that
 the players need can be sampled offline before the protocol starts.
 These properties are used in a security proof presented in \cite{DF12}.

 \begin{center}
 \begin{figure}
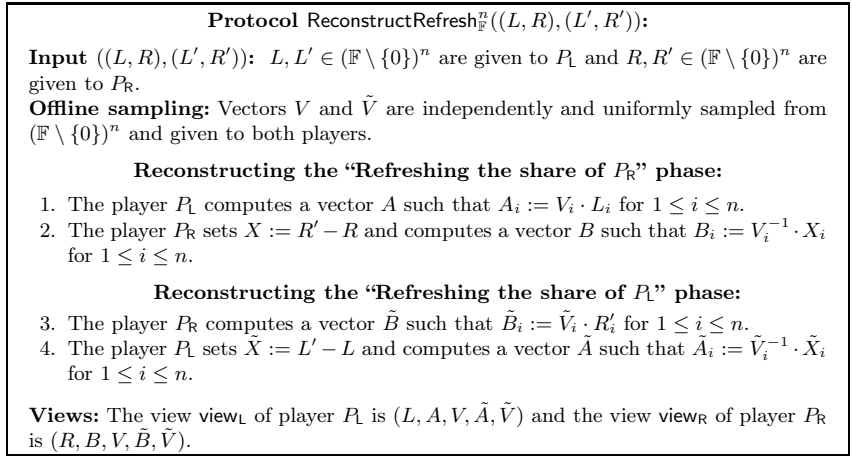

 \scalebox{0.87}{
 \fbox{
 { \footnotesize
 \begin{minipage}[c]{1.0\linewidth}
    \begin{center}
      {\bf Protocol $\ReconstructRefresh^n_\FF((L,R),(L',R'))$: }
    \end{center}
    \vspace{0.2cm}
 \begin{minipage}[c]{1\linewidth}
 {\bf Input $((L,R),(L',R'))$: } $L,L' \in \FZ^n$ are given to $\PL$ and $R, R' \in \FZ^n$ are given to $\PR$.\\
 {\bf Offline sampling:} Vectors $V$ and $\tilde{V}$ are independently and uniformly sampled from $\FZ^n$ and given to both
  players.
  \vspace{-0.1cm}
  \begin{center} {\bf Reconstructing the ``Refreshing the share of $\PR$'' phase:} \end{center}
 \begin{enumerate}
   \vspace{-0.5cm}
    \item The player \label{step:R1Rec} $\PL$ computes a vector $A$ such that $A_i := V_i \cdot L_i$ for $1 \le i \le n$.
    \item The player \label{step:R2Rec} $\PR$ sets $X := R' - R$ and computes a vector $B$ such that $B_i := V^{-1}_i \cdot X_i$ for $1 \le i \le n$.
     \vspace{-0.1cm}
       \begin{center} {\bf Reconstructing the ``Refreshing the share of $\PL$'' phase:} \end{center}
  \vspace{-0.2cm}
  \item The player \label{step:L1Rec} $\PR$ computes a vector $\tilde{B}$ such that $\tilde{B}_i := \tilde{V}_i \cdot R'_i$ for $1 \le i \le n$.
  \item The player \label{step:L2Rec} $\PL$ sets $\tilde{X} := L' - L$ and computes a vector $\tilde{A}$ such that $\tilde{A}_i := \tilde{V}^{-1}_i \cdot \tilde{X}_i$ for $1 \le i \le n$.
 \end{enumerate}
 \vspace{-0.1cm}
 {\bf Views:} The view $\viewL$ of player $\PL$ is $(L,A,V,\tilde{A},\tilde{V})$ and the view $\viewR$ of player $\PR$ is $(R,B,V,\tilde{B},\tilde{V})$.
 \end{minipage}
 \vspace{-0.1cm}
 \end{minipage}
 }
 }}
 \caption{Protocol $\ReconstructRefresh^n_\FF$}
   \label{fig:reconstructrefresh}
 \end{figure}
 \end{center}

 We now formalize what it means that $\ReconstructRefresh_\FF^n$ is a
 reconstructor for $\Refresh_\FF^n$.  This is done by considering two
 experiments depicted on Fig.\ \ref{fig:exp}.  The next lemma shows
 that these experiments produce the same distributions.

 \begin{figure}
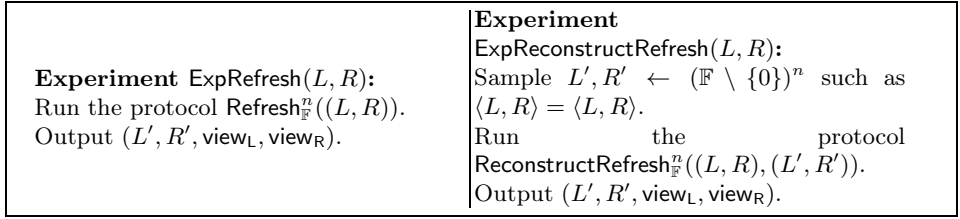

   \centering
 \fbox{
 { \footnotesize
 \begin{minipage}[c]{1.0\linewidth}
   \begin{tabular}[c]{l|l}
    \begin{minipage}[c]{.47\linewidth}
     {\bf Experiment $\ExpRefresh(L,R)$:}\\
     Run the protocol $\Refresh_\FF^n((L,R))$.\\
  Output $(L',R',\viewL,\viewR)$.
    \end{minipage}
 &
    \begin{minipage}[c]{.47\linewidth}
    {\bf Experiment $\ExpReconstructRefresh(L,R)$:}\\
    Sample $L',R' \draw \FZ^n$ such as $\inp{L,R}=\inp{L,R}$.\\
    Run the protocol $\ReconstructRefresh^n_\FF((L,R),(L',R'))$.\\
   Output $(L',R',\viewL,\viewR)$.
    \end{minipage}
  \end{tabular}
 \end{minipage}}}
  \caption{Experiments $\ExpRefresh(L,R)$ and $\ExpReconstructRefresh(L,R)$.}
   \label{fig:exp}
 \end{figure}

 \begin{lemma}\label{lemma:reconref}
   For every $L,R \in \FZ^n$ we have that
 \[
 \ExpRefresh(L,R) \eqd \ExpReconstructRefresh(L,R).
 \]
 \end{lemma}

\begin{proof}
 We only show that the equality of distributions holds for the
 variables involved in the ``Refreshing of the share of $\PR$ phase''
 (the same fact for the other phase is proven analogously).  These
 variables are
\[
L,R,A,B,V,X,R'.
\]

We prove it showing that each of the above variables has an identical conditional
distribution given the previous variables in the series:
\begin{enumerate}
 \item{$\mathbf{L,R}$:} Clearly in both experiments $(L,R)$ is constant and identical;
 \item{$\mathbf{A}$:} $A$ is uniformly distributed over $\FZ^n$ independently of $(L,R)$.
 In the first experiment it comes from the way it is sampled from $\cpO$.
 In the second scenario it is defined by the equation $A_i :=  V_i \cdot L_i$ for $1 \le i \le n$.
Hence, each $A_i$ is a product of $V_i$ distributed uniformly over $\FZ$ and some fixed non-zero $L_i$.
Therefore $A_i$ has a uniform distribution over $\FZ$.
 
  \item{$\mathbf{B}$:} $B$ is uniformly distributed over $\FF^n$ independently of $(L,R,A)$.
 In the first experiment it comes from the way it is sampled from $\cpO$.
 In the second scenario it is defined by the equation $B_i := V^{-1}_i \cdot X_i$ for $1 \le i \le n$.
 Notice that $R'$ has a uniform distribution over $\FF^n$ independent of $(L,R,A)$,
 so $X$ defined by $X := R' - R$ is also uniform over $\FF$.
 Hence, each $B_i$ is a product of some non-zero $V^{-1}_i$
 and $X_i$ distributed uniformly over $\FF$ and independently of $V$.
 Therefore $B_i$ has a uniform distribution over $\FF$.
 \item{$\mathbf{V}$:} $V$ is uniquely determined given $(L,R,A,B)$ by the equation $V_i := L_i^{-1} \cdot A_i$ for $1 \le i \le n$
 (Step \ref{step:R1} in Fig. \ref{fig:refresh} and Step \ref{step:R1Rec} in Fig. \ref{fig:reconstructrefresh}).
 \item{$\mathbf{X}$:} $X$ is uniquely determined given $(L,R,A,B,V)$ by the equation $X_i = V_i \cdot B_i$ for $1 \le i \le n$
 (Step \ref{step:R2} in Fig. \ref{fig:refresh} and Step \ref{step:R2Rec} in Fig. \ref{fig:reconstructrefresh}).
 \item{$\mathbf{R'}$:} $R'$ is in both experiments equal to $L + X$.
\end{enumerate}
\qed
\end{proof}

\subsection*{Acknowledgments}

The author wishes to thank his supervisor Stefan Dziembowski for the guidance in carring out the research and writing this note. 

\small{
\def\shortbib{0}
\bibliography{LRPKC}

\newcommand{\etalchar}[1]{$^{#1}$}
\begin{thebibliography}{DHLAW10}

\bibitem[BKKV10]{DBLP:conf/focs/BrakerskiKKV10}
Zvika Brakerski, Yael~Tauman Kalai, Jonathan Katz, and Vinod Vaikuntanathan.
\newblock Overcoming the hole in the bucket: Public-key cryptography resilient
  to continual memory leakage.
\newblock In {\em FOCS}, pages 501--510, 2010.

\bibitem[DDV10]{DDV10}
Francesco Dav\`{\i}, Stefan Dziembowski, and Daniele Venturi.
\newblock Leakage-resilient storage.
\newblock In {\em Security and Cryptography for Networks, 7th International
  Conference, SCN 2010, Amalfi, Italy, September 13-15, 2010. Proceedings},
  volume 6280 of {\em Lecture Notes in Computer Science}, pages 121--137.
  Springer, 2010.

\bibitem[DF11]{DF11}
Stefan Dziembowski and Sebastian Faust.
\newblock Leakage-resilient cryptography from the inner-product extractor.
\newblock In {\em ASIACRYPT}, pages 702--721, 2011.
\newblock Full version appears on the Cryptology ePrint Archive
  \url{http://eprint.iacr.org/}.

\bibitem[DF12]{DF12}
Stefan Dziembowski and Sebastian Faust.
\newblock Leakage-resilient circuits without computational assumptions.
\newblock In Ronald Cramer, editor, {\em Theory of Cryptography}, volume 7194
  of {\em Lecture Notes in Computer Science}, pages 230--247. Springer Berlin /
  Heidelberg, 2012.

\bibitem[DHLAW10]{DBLP:conf/focs/DodisHLW10}
Yevgeniy Dodis, Kristiyan Haralambiev, Adriana L{\'o}pez-Alt, and Daniel Wichs.
\newblock Cryptography against continuous memory attacks.
\newblock In {\em FOCS}, pages 511--520, 2010.

\bibitem[DP08]{DP08}
Stefan Dziembowski and Krzysztof Pietrzak.
\newblock Leakage-resilient cryptography.
\newblock In {\em FOCS '08: Proceedings of the 49th Annual IEEE Symposium on
  Foundations of Computer Science}, Washington, DC, USA, 2008. IEEE Computer
  Society.

\bibitem[FKPR10]{FaustKPR10}
Sebastian Faust, Eike Kiltz, Krzysztof Pietrzak, and Guy~N. Rothblum.
\newblock Leakage-resilient signatures.
\newblock In Daniele Micciancio, editor, {\em Theory of Cryptography, 7th
  Theory of Cryptography Conference, TCC 2010, Zurich, Switzerland, February
  9-11, 2010. Proceedings}, volume 5978 of {\em Lecture Notes in Computer
  Science}, pages 343--360. Springer, 2010.

\bibitem[FRR{\etalchar{+}}10]{FaustRRTV10}
Sebastian Faust, Tal Rabin, Leonid Reyzin, Eran Tromer, and Vinod
  Vaikuntanathan.
\newblock Protecting circuits from leakage: the computationally-bounded and
  noisy cases.
\newblock In Henri Gilbert, editor, {\em Advances in Cryptology - EUROCRYPT
  2010, 29th Annual International Conference on the Theory and Applications of
  Cryptographic Techniques, French Riviera, May 30 - June 3, 2010.
  Proceedings}, volume 6110 of {\em Lecture Notes in Computer Science}.
  Springer, 2010.

\bibitem[GR10]{GoldwasserR10}
Shafi Goldwasser and Guy~N. Rothblum.
\newblock Securing computation against continuous leakage.
\newblock In Tal Rabin, editor, {\em Advances in Cryptology - CRYPTO 2010, 30th
  Annual Cryptology Conference, Santa Barbara, CA, USA, August 15-19, 2010.
  Proceedings}, volume 6223 of {\em Lecture Notes in Computer Science}, pages
  59--79. Springer, 2010.

\bibitem[GR12]{GR12}
Shafi Goldwasser and Guy~N. Rothblum.
\newblock How to compute in the presence of leakage.
\newblock {\em Electronic Colloquium on Computational Complexity (ECCC)},
  19:10, 2012.
\newblock To be presented on FOCS 2012.

\bibitem[LLW11]{Alison1}
Allison Lewko, Mark Lewko, and Brent Waters.
\newblock How to leak on key updates.
\newblock to appear at STOC 2011, 2011.

\bibitem[MR04]{MR}
Silvio Micali and Leonid Reyzin.
\newblock Physically observable cryptography (extended abstract).
\newblock In Moni Naor, editor, {\em TCC}, volume 2951 of {\em Lecture Notes in
  Computer Science}, pages 278--296. Springer, 2004.

\end{thebibliography}
\bibliographystyle{alpha}}

\end{document}